\documentclass[12pt,preprintnumbers,amsmath,amssymb]{revtex4}
\usepackage[T1]{fontenc}
\usepackage[latin2]{inputenc}
\usepackage{amsfonts}
\usepackage{amsthm}
\usepackage{upgreek}
\usepackage{graphicx}

\newtheorem{tw}{Theorem}
\newtheorem{lem}{Lemma}

\newtheorem{cor}[tw]{Corollary}
\newtheorem{df}{Definition}

\DeclareMathOperator{\sgn}{sgn}

\begin{document}

\title{Game Dynamics for Players \\
 with Social and Material Preferences  }
\author{Tadeusz P{\l}atkowski and Jan Zakrzewski}
\affiliation{Faculty of Mathematics, Informatics, and Mechanics \\
University of Warsaw, Banacha 2, 02-097 Warsaw, Poland}



\begin{abstract}
We consider the dynamics, existence and stability of the equilibrium states   
for large populations of individuals who can play various types of 
non--cooperative games.  
The players imitate the most attractive strategies, and the choice is motivated not only by the material 
payoffs of the strategies, 
but also by their popularity in the population. 
The parameter which determines the weights of both factors in the equilibrium 
states   
has the same analytical form for all types of considered games, and is identified with the 
sensitivity to reinforcements parameter in the Hernstein's Matching Law. 
We prove theorems of existence and uniqueness, and discuss examples of multiple 
locally stable polymorphic equilibria for the considered types of games.    
\end{abstract}

\maketitle

\parindent=0pt

\section{Introduction}


In the evolutionary game theory the theoretical approach to model the 
dynamics of populations is based on the proportional fitness rule. If the random matching 
and imitation of the individuals with highest payoffs are assumed, the  
equations which govern the evolution of such populations are the celebrated replicator equations, 
cf. e.g.  \cite{Veg, Gin1, Szabo, McE,HofSig, Weibull}, and references cited therein.  

However, in general the individuals can be oriented not only toward imitating  the highest  payoff strategies, but 
they can also take into account other, "noneconomic" factors, in particular  
popularity of strategies in the population. 
The idea of combining together the bias towards imitating the strategies of the most payoff-successful agents and the 
strategies of the majority is not new, cf. for example \cite{Hen}, where the ideas of imitating the successful, 
and copying the majority (the conformist transmission), are put together to stabilize the cooperation in 
populations of individuals, see also \cite{Han} for another argument  that noneconomic factors influence human's behavior.    

We consider a theory of evolution of social systems, that generalizes the standard proportional fitness rule of the 
evolutionary game theory. The biological fitness of a behavior, strategy, measured by its payoff from interactions, is 
replaced by more general function, the attractiveness of the behavior. The attractiveness of the strategy is assumed 
to depend not only on its payoff, but also on its actual popularity (fraction) in the population. The parameters of 
the attractiveness function describe different psychological characters of the members of the population. 
 We~consider the model based on the generalized Cobb--Douglas utility function, cf. \cite{Cob}, \cite{Pla1}, \cite{Pla2}. 
We find a parameter  that describes different personality profiles of the players, and that can be identified with 
the sensitivity to reinforcement in the Matching Law of mathematical psychology 
\cite{Her}. In our setting it determines stability of polymorphic equilibria in all 
considered classes of games. We note that such equilibria for games with two 
strategies played in infinite populations were found for 
example in the aspiration--based models, cf. for example \cite{Pal, Pla3}, 
in general multi-person games, cf. \cite{Gok}, in the models of social 
dilemmas with synergy and discounting, cf. \cite{Hau3}, in the Stag-Hunt multi-person games, 
cf. for example \cite{Pac}, in the multi-person Snowdrift game, cf. for example 
\cite{Sou} and references cited therein. 

We prove  theorems of the existence of polymorphic equilibria, and identify sufficient 
 conditions for their uniqueness for particular classes of the considered games. 
We also find examples of the existence of more than one (locally) stable polymorphism.  

In the next section we formulate the general model and discuss its basic properties. In section III we briefly remind, 
for the convenience of the reader, an existence and uniqueness theorem for a general 
class of two-person symmetric games with two strategies, played in populations 
with social and material preferences, proved in \cite{Pla2}. We also provide examples of multiple stable internal equilibria 
for two-person symmetric games with three strategies. In section IV we discuss  polymorphic equilibria for 
two-person asymmetric games with two strategies, and in section V we consider 
the multi--person games, 
in particular the Public Good game. In section VI we conclude and discuss some open problems.

\section{Model}

We consider an infinite homogeneous population of individuals who interact through a 
random matching, playing at each instant of time a two--person or a multi--person 
non--cooperative game. The players have a finite number $K$ of behavioral types 
(strategies). Let $N_i(t)$ denotes the number of individuals playing the strategy 
$i, \ $ $N=N_1+N_2+...N_K$--the fixed size of the system, $p_i=N_i/N \ $--the frequency, 
or popularity of behavior $i$ for $i=1, 2,...,K$. 

The players who play a strategy $i$ review their strategy according to the Poisson process with the arrival rate $r_i$. 
We model the corresponding stochastic processes as a deterministic flow. The balance conservation equations read, 
cf. e.g. \cite{Weibull}, section 4.4:  
\begin{equation}
\dot p_i(t) = \sum_{j \ne i} [p_j r_j p_j^i - p_i r_i p_i^j], \ \ \ i=1,...K,  \label{a11}
\end{equation}
where $p_j^i$ is the probability that the agent playing the $j$ strategy will switch to the $i$ strategy. 

We assume that $p_j^i$ is proportional to the attractiveness $u_i$ of the strategy $i$: 
\begin{equation}
p_j^i=cu_i.  \label{aaa111}
\end{equation}
In general the attractiveness of a strategy can be a complicated function of various 
factors, describing the state of the system and the characteristics of the players. 
Many social and biological interactions are based on imitation processes, where individuals adopt more successful 
and more popular strategies with larger probability than less successful and less popular ones. 
We choose for the attractiveness of the strategy $i$ the generalized Cobb-Douglas utility function, 
cf. \cite{Cob, Pla1, Pla2}:  
\begin{equation}\label{cobb}
u_i(t) = p_i^{1-\alpha}\nu_i^{1-\beta}, \ \ \ i=1,...,K,  
\end{equation}
with $(\alpha, \beta) \in [0,1] \times [0,1].$ The formula (\ref{cobb}) states that the attractiveness of a strategy depends not only on the payoff of the 
strategy, but also on its actual popularity in the population. The parameters $\alpha, \ \beta$ determine the responsiveness of the function $u_i$ to changes of the current 
popularity and of the mean payoff of the action $i$. They define different 
social and material preferences 
of the individuals, in other words their different personality profiles. 

More attractive strategies ought to have an evolutionary advantage in the considered social systems. 
Note that the attractiveness of the strategy $i$ is increasing and concave function of both arguments, 
i.e. of the mean payoff $\nu_i$ and of 
the popularity $p_i$ in the population. In particular when the attractiveness reaches a higher level, the changes 
are slower.  The attractiveness of the strategy becomes zero if its mean payoff or its popularity in the population 
is zero. In the Appendix, p. A we define and characterize the ideal types of the personality profiles for 
$\alpha, \beta \in \{0, 1\}$, and characterize their basic properties, cf. 
\cite{Pla1, Pla2}. 
All other values of the parameters $\alpha, \beta$ describe intermediate personality profiles. 
We show below that the combination of these parameters
\begin{equation}\label{s}
s=\frac{1-\beta}{\alpha},
\end{equation}
plays a crucial role in determining the polymorphic equlibria and their stability,  
and can be identified with the sensitivity parameter which links the relative rates 
of reinforcements and responses in the Hernstein's Matching Law \cite{Her}. 

As postulated in (\ref{aaa111}), the strategies with higher attractivenesses have larger probability to be imitated. Assuming that the arrival 
rates $r_j$ are constants, and rescaling the time we obtain the system of equations which has the formal structure 
analogous to the replicator equations: 
\begin{equation}
\dot p_i(t) = u(\frac{u_i}{u} - p_i), \ \ \ \ \ u:=\sum_{j=1}^K u_j, \  \ i=1,2,...,K. \label{a2}
\end{equation}
Thus, the fraction $p_i$ of strategy $i$ increases if its normalized attractiveness $\frac{u_i}{u} \in [0, 1]$ 
is bigger than the actual fraction of the strategy $i$, and decreases if it is 
smaller. In particular, for $\alpha = \beta = 0$, corresponding to Homo Afectualis (cf. the Appendix, p. A) the evolution equations 
(\ref{a2}) are identical to the replicator equations of the evolutionary game theory. 

All critical points of the dynamics (\ref{a2}) are obtained as solutions of the 
 system of $K-1$ algebraic equations
\begin{equation}
\frac{u_1}{p_1} = \frac{u_2}{p_2}=...=\frac{u_K}{p_K}.  \label{a6001}
\end{equation}
which, after substituting (\ref{cobb}) is equivalent to 
$$
\frac{p_i}{p_1} = \left(\frac{\nu_i}{\nu_1}\right)^s,  \quad   i=1,...,K,  \label{a6002}
$$
where $\nu_i \ $ is the mean payoff of strategy $i$. In particular the stability 
properties of the solutions of eqs. (\ref{a6001}) 
depend on the combination $s$ of the parameters $\alpha, \beta$, which characterize the personality 
profile of the players, rather than separately on each of them. The sensitivity parameter $s$ plays an important 
role in the matching law in the operant response theory of the mathematical psychology, in particular as 
a measure of the degree to which, in equilibrium, the response ratio changes when the reinforcement ratio is modified,  
cf. for example \cite{Her, Pla1, Pla2}, and references cited therein.

\section{Equilibria for Two-Person Symmetric Games}

In this section we consider populations which play symmetric 2-person games. For 
a convenience of the reader we first remind the results for the case of two strategies, 
proved in \cite{Pla2}. 

\subsection{2-Person Games with 2 Strategies}\label{model2ch}

For the symmetric 2-person games with two strategies (denoted 1,2), with the  payoff matrix 
\begin{equation}\label{aaa}
\begin{array}{r|cc} 
        & 1 & 2 \\ \cline{1-3} \hline 
        1 & a & b \\ 
        2 & c & d
\end{array} 
\end{equation}
where $a, b, c, d$  are arbitrary positive numbers, we obtain the 
full characterization of the equilibria. With the normalization condition $p_1+p_2=1$ eq. (\ref{a2}) reduces to the evolution equation 
\begin{equation}
\dot p_1=(1-p_1)^{1-\alpha}p_1^{1-\alpha}[\nu_1^{1-\beta}(1-p_1)^\alpha-\nu_2^{1-\beta}p_1^\alpha],  \label{a300}
\end{equation}
with $\nu_1=(a-b)p_1+S, \ \ \nu_2=(c-d)p_1+d. $ For each personality profile $(\alpha,\beta)\in [0,1] \times [0,1]$ 
there exist two pure equilibria 
of (\ref{a2}): $p_1=0$ and $p_1=1$ [for $\alpha=1$ we put $u_i(p_i=0)=0, \ i=1,2$]. The equilibria of eq. (\ref{a2}) in 
which each strategy has a non zero frequency  will be called {\bf{Mixed Equilibria}}, 
and denoted ME. Each ME corresponds to a fixed point of eq. (\ref{a300}). 
For the symmetric 2-person games the following theorem is true:  

\begin{tw} $ \ $

For the payoff matrix (\ref{aaa}) with positive entries:  

\vskip 0.2cm
I. For each $ 0 \le s < \infty$ there exists at least one ME -- the fixed point of
 the evolution equation (\ref{a300}).  

\vskip 0.1cm 
II. Denote  
$B:=(1-s)\frac{d}{c} + (1+s)\frac{b}{a}, \ \ \Delta:=B^2-4\frac{bd}{ac}. $ 
If 1. $\Delta \le 0,$ or 2. $B \ge 0$, or 3. $bc \ge ad$, then ME is unique.   

\vskip 0.1cm
III. For each $ 0 \le s < \infty$ there exist at most three ME. 

\vskip 0.1cm
IV. There  exist three ME iff 
$\Delta > 0, \ \ B<0, \ \mbox{and} \  U(z_1) U(z_2)<0, $ where 
$z_{1,2} := [-B \mp \sqrt{\Delta}]/2,$ and  
$U(z):=ln z +s \ ln{\frac{cz+d}{az+b}}, \ \ z>0. $  

\vskip 0.1cm
V. If a ME is unique, then it is globally stable in $(0, 1)$ under the considered dynamics. If 
there are three ME, then the middle one is (locally) instable, the other two stable. 
\end{tw}

\vskip 0.1cm
The statement I implies that the dependence of the attractiveness of a strategy on its popularity  
implies the existence of at least one ME for all symmetric two-person games with positive payoffs, 
including all types of the social dilemma games. In particular this "solves the dilemma"  of the Prisoner's 
Dilemma population game. The statement II.2 implies that for $s \le 1$ the cooperation level is 
unique, and does not depend on the initial distribution of strategies.  
For the proof, interpretations and applications to various types of 2-person games 
the reader is referred to \cite{Pla2}.

\subsection{2-Person Symmetric Games with 3 Strategies}\label{model2ach}

\subsubsection{General results}

2-person symmetric games with 3 strategies play important role in the population game theory. 
First we formulate the general setting in the frame of populations with complex personality profiles, 
then we discuss particular types of such games.  

The general symmetric 2-person game with 3 strategies, denoted 1, 2, 3, has the payoff matrix 
\begin{equation}\label{3matrix}
	\begin{array}{r|ccc}
		& 1 & 2 & 3 \\ \hline
		1 & a_{11} & a_{12} & a_{13} \\
		2 & a_{21} & a_{22} & a_{23} \\
		3 & a_{31} & a_{32} & a_{33}
	\end{array}
\end{equation}
where we assume $a_{ij}>0 \ \forall i,j \in \{1, 2, 3\}.$ The mean payoffs of strategies $1,2,3$ are defined as
\begin{align}\label{3mean} 
\nu_i(t) &= \sum_{j=1}^{3} a_{ij}p_j(t),  \quad i=1,2,3,
\end{align}
with $p_3=1-p_1-p_2$. The dynamical system (\ref{a2}) has the form 
\begin{align}\label{2dynamic} 
\dot p_1 &= p_1^{1-\alpha}(1-p_1)\nu_1^{1-\beta}-p_1[p_2^{1-\alpha}\nu_2^{1-\beta}+(p_3)^{1-\alpha}\nu_3^{1-\beta}],   \\
\dot p_2 &= p_2^{1-\alpha}(1-p_2)\nu_2^{1-\beta}-p_2[p_1^{1-\alpha}\nu_1^{1-\beta}+(p_3)^{1-\alpha}\nu_3^{1-\beta}]. 
\end{align}
For the 3-strategy games we define mixed equilibria (ME) of (\ref{2dynamic}) as the critical points of 
the above dynamical system with all nonzero coordinates $p_i, \ i=1, 2, 3$. We shall alternatively use the term 
{\it{polymorphic equilibria}}. With notation 
\begin{equation}
x:=\frac{p_2}{p_1}, \quad y:=\frac{p_3}{p_1}
\end{equation}
the equations for ME can be written, using (\ref{cobb}), (\ref{a6001}), (\ref{3mean}), in the form 
\begin{align}
x &=\left(\frac{a_{21}+a_{22}x+a_{23}y}{a_{11}+a_{12}x+a_{13}y} \label{alg1}\right)^s,  \\
y &=\left(\frac{a_{31}+a_{32}x+a_{33}y}{a_{11}+a_{12}x+a_{13}y} \label{alg2}\right)^s. 
\end{align}
The frequencies of the strategies in ME are obtained from the reverse formulas: 
\begin{equation}\label{revers}
p_1=\frac{1}{1+x+y}, \ \ p_2=p_1 x, \ \ p_3=p_1 y. 
\end{equation}
In general the number of ME and their stability depend on the payoff matrix and on the sensitivity parameter $s$. 

One of the most celebrated 3-strategy games is the Rock-Paper-Scissor (RPS) game, cf. for example \cite{HofSig,Weibull}. 
The population of players with complex personality profiles, playing the RSP game 
has been investigated in \cite{Pla4}. In particular, in the "standard" RPS game the critical point 
$(\frac{1}{3}, \frac{1}{3}, \frac{1}{3})$ is asymptotically~stable for all  $\alpha, \beta  \in (0,1]$. 
For other properties of the solutions, in particular the existence of limit cycles, the Hopf and the Boutin--like 
bifurcations, the reader is referred to \cite{Pla4}.     

Below we analyze other important classes of 3--strategy symmetric games, in particular the coordination games, 
for which there exist multiple stable polymorphic equilibria, and the iterated Prisoner's Dilemma game with three strategies 
AllD, AllC, TFT, with unique stable ME.

\subsubsection{Coordination Game}

We consider the coordination game with 3 strategies and the payoff matrix 
\begin{equation}\label{33matrix}
	\begin{array}{r|ccc}
		& 1 & 2 & 3 \\ \hline
		1 & 2 & 1 & 1 \\
		2 & 1 & 3 & 1 \\
		3 & 1 & 1 & 3
	\end{array}
\end{equation}
We demonstrate the existence of three ME for $s=3$.  
For $s=1$ the system (\ref{alg1}), (\ref{alg2}) can be solved explicitly. 
The unique nonnegative solution is $x=y=\frac{1}{2}(1+\sqrt{3}),$ which, according to (\ref{revers}), 
gives the equilibrium state $(p_1, p_2, p_3)$ with the frequencies  
$p_1 = 2-\sqrt{3}, \ \ p_2 = p_3 = \frac{1}{2}(\sqrt{3}-1).$ 

For $s=2$ the system (\ref{alg1}), (\ref{alg2}) is of the third order. In order to obtain analytical solution  
we subtract (\ref{alg2}) from (\ref{alg1}), and obtain the equation 
\begin{equation}\label{rrr}
(x-y)(2+x+y)^2=4(x-y)(1+2x+2y).
\end{equation}
It can be easily shown that (\ref{rrr}) has positive solution only if $y=x$, in which case 
the solution of the system (\ref{alg1}), (\ref{alg2}) is the pair $(x_0,x_0)$, where $x_0\cong 2.45$ is the unique positive 
root of the polynomial 
\begin{equation}
F(x)=4x^3-8x^2-4x-1.  
\end{equation}
The corresponding unique equilibrium state reads $(p_1, p_2, p_3)\cong ( 0.17, 0.415, 0.415).$ 
When we farther increase the value of the sensitivity parameter $s$ the uniqueness is lost. 
In Figure \ref{zak1} we show time evolution of frequencies ($p_1,p_2,p_3$) for two 
values of the sensitivity parameter $s$ for particular initial data, symmetric with respect to $p_{10}$ 
(i.e. for $p_{20}=p_{30}$). For $s=2.15$ the equilibrium is unique (not shown). 
For a value $s \ge s_0 \in (2.15, 2.30) $ there emerge multiple equilibria. For $s=3$ (left diagram), and $s=3.75$ 
(right diagram) there are three ME, two of them stable, one unstable. In 
each diagram the frequencies reach asymptotically one of the stable ME. 
Note the influence of the unstable equilibrium in the early stage of the evolution. First  
the trajectory approaches the unstable symmetric ME 
$(p_1, p_2, p_3)\cong (0.14, 0.43, 0.43)$, then it 
tends to the locally stable equilibrium:  $(p_1, p_2, p_3)\cong (0.126, 0.225, 0.649)$. 
The other "symmetric" locally 
stable ME is $(p_1, p_2, p_3)\cong ( 0.126, 0.649, 0.225)$.  The amount of time the trajectory stays in a neighborhood 
of the unstable equilibrium depends on the sensitivity parameter $s$, and on the personality profile 
parameters $\alpha,  \beta$. For $s=3.75$ both locally stable ME tend to the relevant 
vortexes of the game simplex, cf. the right diagram in Figure \ref{zak1}. 

We explain the emergence of the multiple equilibria for $s=3$.  
First we look for the symmetric solution $p_2 = p_3$, i.e. $x=y$. The system (\ref{alg1}), (\ref{alg2})   
reduces to one equation
\begin{equation}\label{funkcjaFsym}
8x^4-40x^3-24x^2-4x-1=0.
\end{equation}
The unique positive root of the polynomial in (\ref{funkcjaFsym}) is $x \cong 5.556$. The corresponding 
equilibrium reads $(p_1, p_2, p_3) \cong (0.082, 0.459, 0.459)$. 
There are two other, symmetric  ME. The first zero gives the ME:    
$(p_1, p_2, p_3) \cong (0.047, 0.057, 0.896)$. As expected from the symmetry of the second and third strategy, the 
second zero gives the "symmetric" ME: $(p_1, p_2, p_3) \cong (0.047, 0.896, 0.057)$. There are also four "partially mixed" 
equilibria: two on the edge $p_1=0$, and one on each other edges, in the neighborhood of the corners 2 and 3, i.e. 
near $p_2=1$ and $p_3=1$. For increasing values of the sensitivity parameter $s$ the stable 
equilibria approach the relevant edges of the game simplex. 
\begin{figure}[!hbt]
\includegraphics[width=0.35\textwidth,angle=-90]{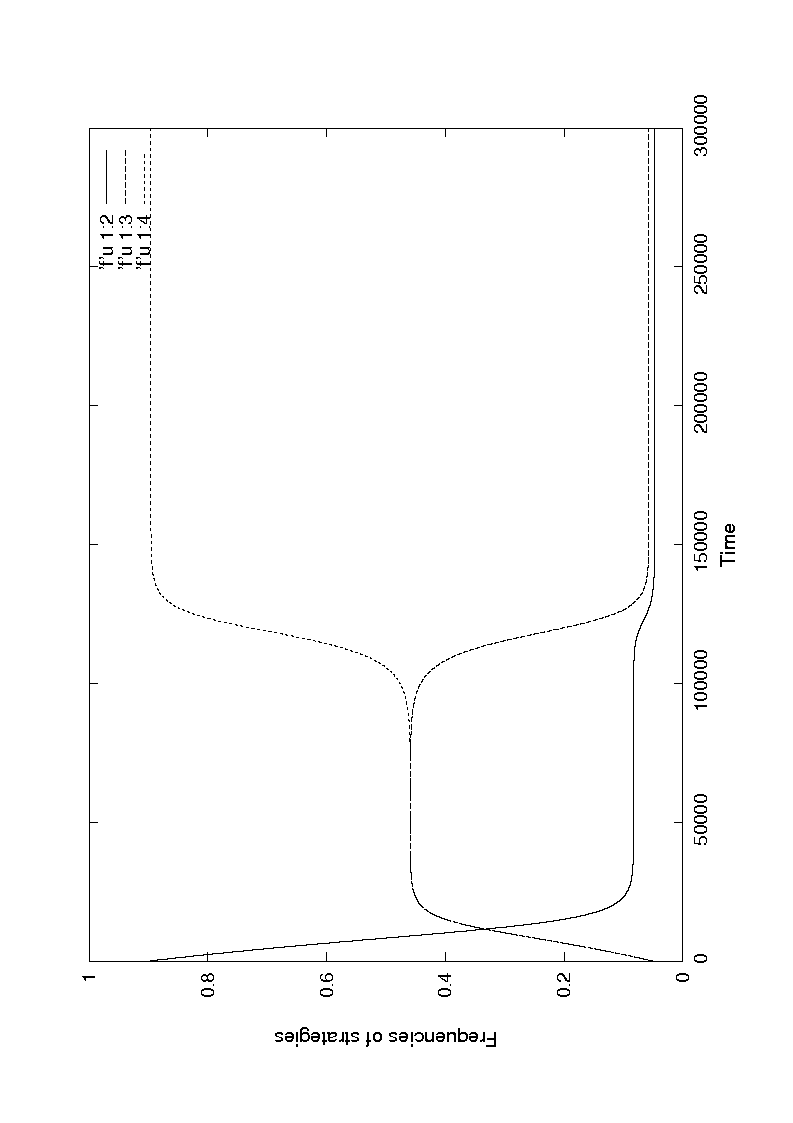}
\includegraphics[width=0.35\textwidth,angle=-90]{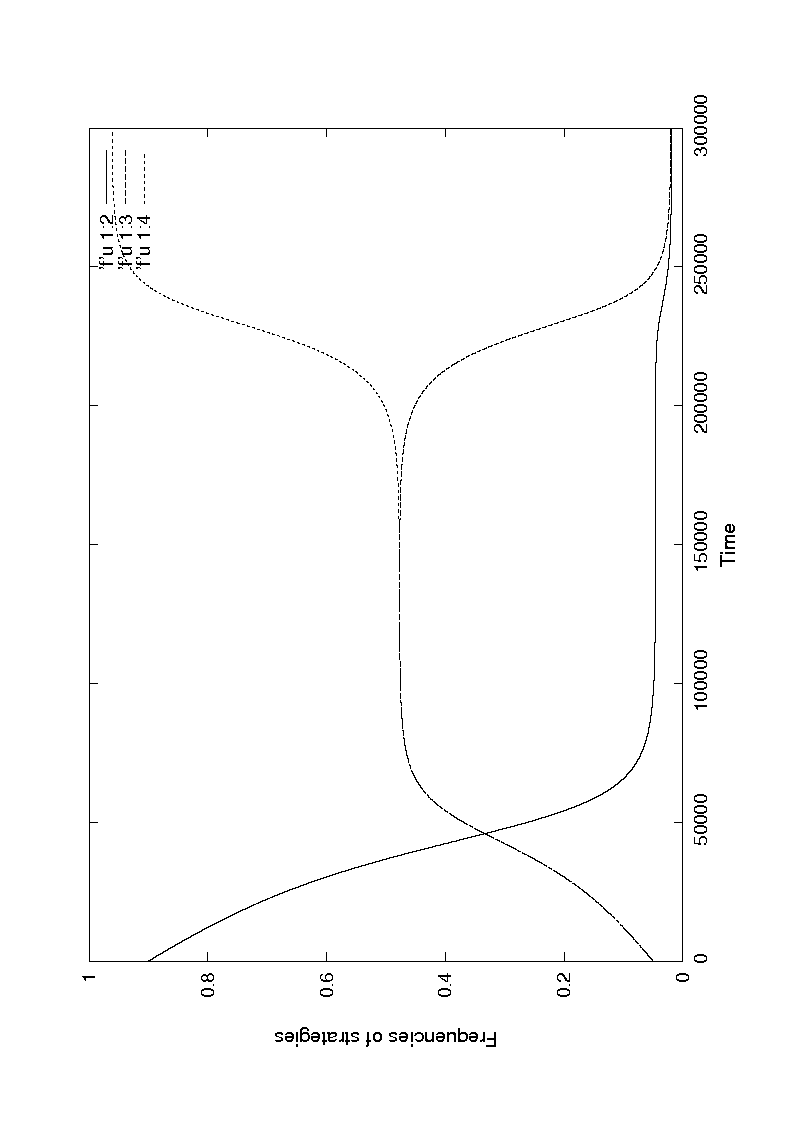}
\caption{Time evolution of frequencies $p_1,p_2,p_3$ for the coordination game (\ref{33matrix}). Left: $s=3$, right: $s=3.75$. \label{zak1}}
\end{figure}
In the Appendix, p. B we prove the existence and uniqueness theorem for another class of the 
asymmetric coordination games, 
in which both players are better off if they play different strategies (anti--coordination games).

\subsubsection{Repeated Prisoner's Dilemma Game}

Another interesting example of three--strategy games is the 2--person, ~3--strategy finitely repeated Prisoner's Dilemma 
Game with the payoff matrix 
\begin{equation}\label{3matrix_PD}
	\begin{array}{r|ccc}
		& AllC & AllD & TFT \\ \hline
		AllC & Rm & Sm & Rm \\ 
		AllD & Tm & Pm & T+P(m-1) \\
		TFT & Rm & S+P(m-1) & Rm
	\end{array}
\end{equation}
where $T>R>P\ge S$ are the payoffs in the one-shot PD game, AllC (AllD) is the strategy:
 play always C (play always D), TFT is the strategy Tit for Tat, and $m$ is the number of rounds.  
 
In the classical replicator dynamics, with the payoff from TFT additionally reduced by a small value 
(the cost of playing TFT) it has been shown 
in \cite{FudNow}, that the population evolves through the cycles of cooperation and 
defection. 

For the general matrix (\ref{3matrix_PD}) we checked numerically the existence and uniqueness of ME for 
for wide ranges of the payoffs $T> R> P \ge S$ and of the parameters $m, s$. 
Moreover, the share (sum of the frequencies) of the cooperative strategies  AllC 
and TFT decreases for increasing ratio $\frac{T}{R}$ and increasing number 
of rounds, whereas for increasing $m$ it increases, in agreement with intuition.  

For the Weak Prisoner's Dilemma ($P=S=0$) we prove 
\begin{cor} $ \ $

There exists an unique ME: 
$(p_1, p_2, p_3)=(\frac{1}{2+a}, \frac{a}{2+a}, \frac{1}{2+a})$, 
where $a=[\frac{T(m+1)}{2mR}]^s $, 
of the game (\ref{3matrix_PD}) with $P=S=0$, for all nonnegative sensitivity parameters $s$. 
\end{cor} 
\begin{proof}
Eq. (\ref{alg2}) for the ME gives $y=1$. Eq. (\ref{alg1}) gives 
$a=[\frac{T(m+1)}{2m}]^s$, 
and the result follows from (\ref{revers}).
\end{proof}
Note that for the increasing ratio $\frac{T}{R}$ the shares of the AllC and TFT players in ME decrease, whereas for 
increasing $m$ they increase, as in the case of the general payoff matrix (\ref{3matrix_PD}).

\section{2-Person Asymmetric Games}\label{model3ch}

\subsection{General Results}

In this section we consider 2-person games played between members of two populations with (in general) different 
personality profiles. The members of both populations choose between two strategies $A$ and $B$. The payoff matrix reads 
\begin{equation}
	\left[\begin{array}{cc} (a_1,a_2) & (b_1,b_2) \\ (c_1,c_2) & (d_1,d_2)\end{array}\right], 
\end{equation}
with nonnegative entries. In general the strategies in both populations may be different meanings. 
For simplicity we do not use different notation for the relevant pairs of strategies. 
The row players belong to the population~$i=1$, and the 
column ones to ~$i=2$. We denote  $x_i,\ i=1,2,$ the fraction of population $i$ which plays strategy $A$. 

Let $(\alpha_i,\beta_i)$ describe the personality profile of the individuals of population $i, \ i=1.2$. The members of 
each population may put different weights to the payoffs and popularities of the available strategies in the imitation process.  

We define $u_j^i$--the attractiveness of strategy $j \in \{A,B\}$ in population $i \in \{1,2\}$: 
\begin{eqnarray} \label{attra}
	u_A^i & = & x_i^{1-\alpha_i}\nu_{Ai}^{1-\beta_i},\\ u_B^i & = & (1-x_i)^{1-\alpha_i}\nu_{Bi}^{1-\beta_i},
\end{eqnarray}
where $\nu_{ji}$ is the mean payoff from strategy $j$ in population $i$: 
$$\nu_{A1}=a_1x_2+b_1(1-x_2), \ \  
\nu_{B1}=c_1x_2+d_1(1-x_2), \ \ \nu_{A2}=a_2x_1+c_2(1-x_1), \ \nu_{B2}=b_2x_1+d_2(1-x_1).$$ 
Thus, the attractiveness of each strategy of a population depends on two factors: the ''social'' one, 
represented by the fraction, popularity of the strategy in the population, and the ''economic'' one, represented by the 
mean payoff of this strategy in the considered population. Note that the social factor depends on the composition of 
the considered population, whereas the economic one depends on the composition of the second population.  
This can be interpreted in the following way: the members of each population determine their strategy choice  
observing popularity of the available strategies in their own population, and their payoffs from the 
interactions with the members of the other population. 

The dynamics (\ref{a2}), written for asymmetric games, reads 
\begin{equation}\label{systempocz}
\begin{array}{rcl}
	\dot{x}_1 &=& (1-x_1)x_1^{1-\alpha_1}\nu_{A1}^{1-\beta_1} - x_1(1-x_1)^{1-\alpha_1}\nu_{B1}^{1-\beta_1}, \\
	\dot{x}_2 &=& (1-x_2)x_2^{1-\alpha_2}\nu_{A2}^{1-\beta_2} - x_2(1-x_2)^{1-\alpha_2}\nu_{B2}^{1-\beta_2}.
\end{array}
\end{equation}
Note that for $\alpha_i=\beta_i=0, i=1,2$ the system (\ref{systempocz}) reduces to the replicator dynamics for 
asymmetric games, cf. for example \cite{HofSig, Weibull}. The pairs of $(x_1, x_2)$: $(0,0), (0,1), (1,0), (1,1)$  are  equilibria in pure strategies of (\ref{systempocz}). They 
correspond to the situations in which each population plays only one of two available 
strategies. We are looking for the mixed equilibria (ME) in which both strategies have 
nonzero frequencies for each population.  

With the substitution $z_i=\frac{x_i}{1-x_i}, \ i=1,2$ the system reads
\begin{equation}\label{dyn2z}
	\begin{array}{rcl}
	\dot{z_1} &=& z_1^{1-\alpha_1}(1+z_1)^{\alpha_1}\nu_{A1}^{1-\beta_1} - z_1(1+z_1)^{\alpha_1}\nu_{B1}^{1-\beta_1}, \\[2mm]
	\dot{z_2} &=& z_2^{1-\alpha_2}(1+z_2)^{\alpha_2}\nu_{A2}^{1-\beta_2} - z_2(1+z_2)^{\alpha_2}\nu_{B2}^{1-\beta_2}.
	\end{array}
\end{equation}
From~(\ref{dyn2z}) we obtain the algebraic equations for ME:
\begin{equation}\label{zera2z}
	\begin{array}{rcl}
	\bar{z}_1 = \left(\frac{a_1\bar{z}_2+b_1}{c_1\bar{z}_2+d_1}\right)^{s_1}, \quad
	\bar{z}_2 = \left(\frac{a_2\bar{z}_1+b_2}{c_2\bar{z}_1+d_2}\right)^{s_2},
	\end{array}
\end{equation}
where  
\begin{equation}\label{2sens}
s_i=\frac{1-\beta_i}{\alpha_i}, \quad i=1,2
\end{equation}
denotes sensitivity of players in population $i$. Thus, for the asymmetric 2-person games the ME  
are determined by relevant combinations of parameters $s_1, s_2$ which describe the 
personalities of players respectively in the first and the in the second population. 
In the next subsections we show the existence of ME for typical asymmetric games. 
Here we prove that, contrary to the replicator equations for asymmetric 2-person games, the solutions 
can not be periodic.  
\begin{tw}\label{nieokr}
The system (\ref{dyn2z}) does not have periodic solutions. 
\end{tw}
\begin{proof}[Proof]
We define the function $\phi(z_1,z_2)=[z_1(1+z_1)^{\alpha_1}z_2(1+z_2)^{\alpha_2}]^{-1}$, and note that 
	$$
		\frac{\partial(\phi \dot{z}_1)}{\partial z_1} + \frac{\partial(\phi \dot{z}_2)}{\partial z_2} = -\frac{\alpha_1z_1^{-(\alpha_1+1)}\nu_{A1}^{1-\beta_1}}{z_2(1+z_2)^{\alpha_2}} - \frac{\alpha_2z_2^{-(\alpha_2+1)}\nu_{A2}^{1-\beta_2}}{z_1(1+z_1)^{\alpha_1}} < 0,
	$$
	for $z_1, z_2>0$. Applying the criterion ~\emph{Dulac--Bendixson} we obtain the thesis. 
\end{proof}

\subsection{Asymmetric coordination games}

\subsubsection{Pure coordination games}

We consider the model defined in the previous section for the 
asymmetric coordination games with the payoff matrix 
\begin{equation}
	\left[\begin{array}{cc} (a_1,a_2) & (0,0) \\ (0,0) & (d_1,d_2)\end{array}\right], \label{coordmatrix}
\end{equation} 
and $a_i \neq 0, \ d_i \neq 0, \ i=1, 2,$ played between two populations with 
different personality profiles. Such games will be called pure (asymmetric) coordination games. 
We can for example think about two populations of different sexes, 
$i=1$ and~$i=2$ denote respectively men's and woman's populations. For $a_1 > d_1, \ a_2 < d_2$ this corresponds to the 
well known Battle of Sexes game. Players from each population evaluate 
the available strategies according to their attractiveness, calculated for each of 
two populations $i=1,2$ with the personality profiles  ($\alpha_i, \beta_i$),  
according to formulas (\ref{attra}). With notation (\ref{2sens}) we prove the following theorem
\begin{tw} $ \ $ \label{twistnienie}

If $s_1s_2\neq 1$ then there exists a unique ME of the dynamics (\ref{dyn2z}) for the populations 
with the payoff matrices (\ref{coordmatrix}): 
\begin{equation}\label{stackoord}
	\bar{z}_1 = \left(\frac{a_1a_2^{s_2}}{d_1d_2^{s_2}}\right)^{\frac{s_1}{1-s_1s_2}}, \quad 
    \bar{z}_2 = \left(\frac{a_1^{s_1}a_2}{d_1^{s_1}d_2}\right)^{\frac{s_2}{1-s_1s_2}}.
\end{equation}	
The equilibrium (\ref{stackoord}) is locally asymptotically stable if $s_1s_2<1$, and unstable if $s_1s_2>1$.

If $s_1s_2 = 1$ then for $ (a_1/b_1)^{s_1} a_2^{s_2} \neq 1$ there are no ME, for  
$ (a_1/b_1)^{s_1} a_2^{s_2} = 1$  there exists the family of ME 
$\bar{z}_1 = c, \quad \bar{z}_2 = c^{s_1} (a_2/b_2)^{s_2}$, where $c$ is an arbitrary positive constant. 
\end{tw}
Proof: cf. the Appendix, p. C. Numerical simulations indicate that if the equilibrium (\ref{stackoord}) is locally stable then it attracts all 
investigated trajectories from the interior of the simplex of the game. 
We note that  the replicator dynamics for the considered asymmetric games 
(which formally coresponds to 
infinite sensitivities in our dynamics, has an unique ME, which is Lapunov stable. 
In Appendix, p. D, we prove analogous existence and uniqueness theorem for another 
class of asymmetric games (anti--coordination games), in which the players are 
better off if they use different strategies.

\subsubsection{Coordination games with multiple ME}\label{ex_2chkoord}

\begin{figure}[!hbt]
	\includegraphics[width=0.5\textwidth,angle=-90]{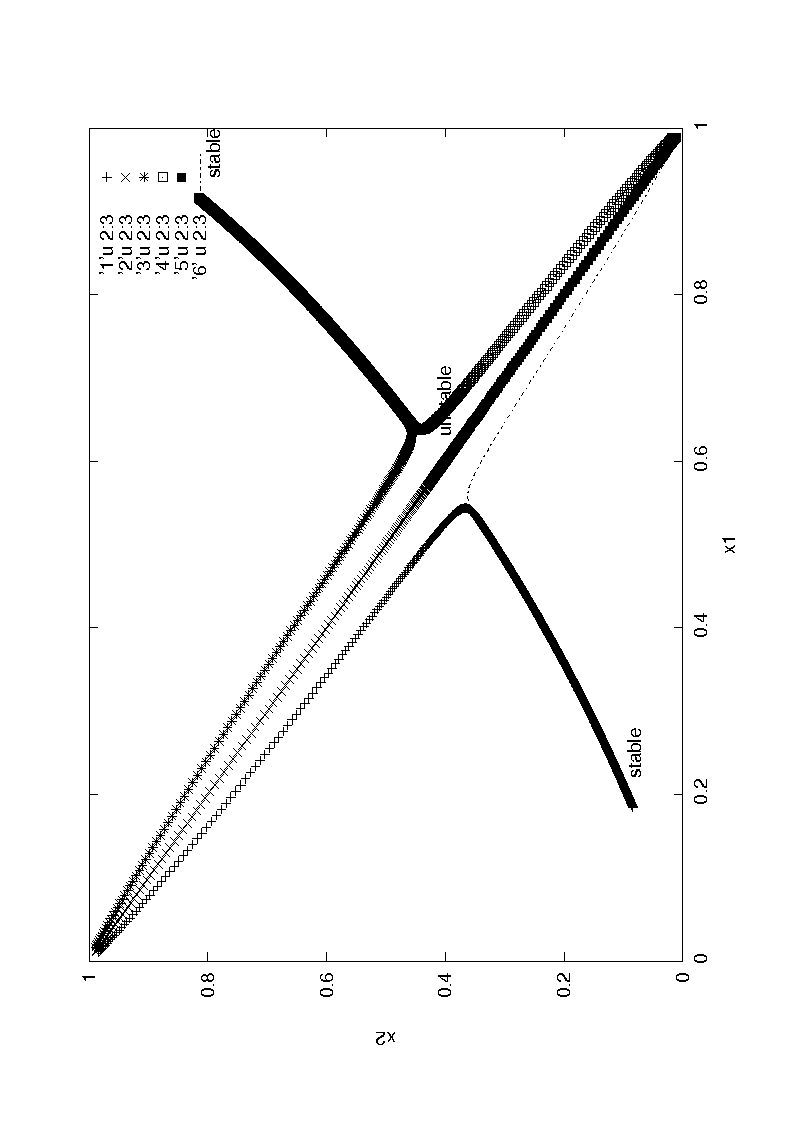}
\caption{Examples of trajectories and critical points for the asymmetric coordination game (\ref{aa1}): 
$a_1=d_2=3, a_2=d_1=2, b_i=c_i=1, i=1,2. $ All six displayed trajectories start from a neighborhood of 
the points $(0,1)$ or $(1,0)$. 
\label{coord1111111}
}
\end{figure}
In asymmetric coordination games the positivity of payoffs plays crucial role in 
the uniqueness problem, and in determining the polymorphic equilibria. In the pure 
coordination games considered above the ME is unique, whereas in general 
coordination games multiple stable polymorphic equilibria are possible.  
For example, the 2-person asymmetric game with the payoff matrix 
\begin{equation} \label{aa1}
	\left[\begin{array}{cc} (3,2) & (1,1) \\ (1,1) & (2,3)\end{array}\right], 
\end{equation}
has, for $s_1=s_2=3$, three ME, corresponding to 
$(x_1 \approx 0.92, x_2 \approx 0.82)$, $(x_1 = 0.6, x_2 = 0.40)$, 
and $(x_1 \approx 0.18, x_2 \approx 0.08)$, as can be checked solving (\ref{zera2z}). Two equilibria are locally stable, one unstable, 
cf. Fig. \ref{coord1111111}.

\section{Multi-person Games}

\subsection{General results} 

In this section we consider multi-person one-shot symmetric games in which each 
of $N$ players  
chooses one of two actions: C or D, C stands for cooperation, 
D for defection. In order to obtain symmetric notation and formulas, 
we define $n:=N-1$, and from now on we consider the $n+1$--person games. 
This allows to define: 
$$
	\left\{\begin{array}{cl}
		a_k: & \textrm{the payoff of player \ } C\textrm{, when k others play \ } C, \\
		b_k: & \textrm{the payoff of player \ } D\textrm{, when k others play \ } D,
	\end{array} \right.  \label{wyplatynnn}
$$
where $k \in \{0,\ \ldots ,\ n\}, $ and $a_k, \ b_k \geq 0$. 

Let $x$ denote the fraction of the population that plays $C$. The mean payoffs from both strategies are respectively
\begin{eqnarray*}
	\nu_C & = & \sum_{k=0}^{n} {n \choose k} x^k(1-x)^{n-k} a_k, \\
	\nu_D & = & \sum_{k=0}^{n} {n \choose k} x^k(1-x)^{n-k} b_k.
\end{eqnarray*}
 
\noindent The evolution equation reads 
\begin{equation}\label{dynnosx}
	\dot{x} = (1-x)x^{1-\alpha}\nu_A^{1-\beta} - x(1-x)^{1-\alpha}\nu_B^{1-\beta},
\end{equation}
Substituting $z=\frac{x}{1-x}$ we obtain
\begin{equation}\label{dynnos}
	\dot{z} = f(z)\left(\left(\frac{W_A(z)}{W_B(z)}\right)^{1-\beta}-z^{\alpha}\right),
\end{equation}
where $W_A(z) = \sum\limits_{k=0}^n{n\choose k}a_kz^k$, $W_B(z) = \sum\limits_{k=0}^n{n\choose k}b_kz^k$, and ~$f(z)$ is a positive function which does not influence the equilibria and their stabilities but only the speed of the evolution. 

Stationary points of the dynamics described by ({\ref{dynnos}}) can be obtained solving the equation 
$$
	\bar{z} = \left(\frac{W_A(\bar{z})}{W_B(\bar{z})}\right)^s,
$$
where $s=\frac{1-\beta}{\alpha}$ is the previously defined sensitivity coefficient. We prove the following
\begin{tw}$  \ $ \label{maxstac}

The dynamics (\ref{dynnosx}) of the symmetric $(n\!+\!1)$--person game has at most 
$2n+1$ mixed equilibria.
\end{tw}
\begin{proof}[Proof]

We define the function $U(z)\colon \mathbb{R}_+ \to \mathbb{R}_+$:  
$$
	U(z) = s\ln{\frac{W_A(z)}{W_B(z)}} - \ln{z}.
$$
Zeros of $U$ are stationary points of the considered dynamics, and  $\sgn(U(z)|_{z_0})=\sgn(\dot{z}|_{z_0})$.

It is sufficient to show that $U(z)=0$ at at most $2n+1$ points in $(0,+\infty)$. To this end we calculate
	$$
		U'(z) = s\frac{W_A'}{W_A} - s\frac{W_B'}{W_B} - \frac{1}{z} = \frac{sW_A'W_Bz - sW_B'W_Az - W_AW_B}{W_AW_Bz}.
	$$
The polynomial $sW_A'W_Bz - sW_B'W_Az - W_AW_B$ is at most of the order $2n$, therefore $U'(z)$ has at most $2n$ zeros. Since between each two zeros of $U$ there has to be a zero of $U'(z)$, the function $U(z)$ can have at most $2n+1$ zeros.
\end{proof}

\begin{tw}$ \ $\label{wnstab0}

For the (n+1)--person games with two strategies and with all payoffs positive there exists at least one mixed equilibrium. 
\end{tw}
\begin{proof}
Due to continuity of the considered dynamics is enough to prove that both boundary equilibrium points corresponding to pore equilibria: $x=0, \  x=1$ are unstable. This will be proved in two lemmas below. 
\end{proof}
\begin{lem}\label{stab0}
	Let $p = \min{\{ i\colon a_i>0\}},\ r = \min{\{ i\colon b_i>0\}}$. If
	\begin{eqnarray}
	\label{war1a} ps>rs+1, & \textrm{ or}& \\
	\label{war1b} ps=rs+1 & \textrm{ and } & {n \choose p}a_p<{n \choose r}b_r,
	\end{eqnarray}
	then the point $x=0$ is asymptotically stable. If 
	\begin{eqnarray*}
		ps<rs+1, &\textrm{ or }&\\
		ps=rs+1 &\textrm{ and }& {n \choose p}a_p>{n \choose r}b_r,\\
	\end{eqnarray*}
	then the point $x=0$ is unstable.
\end{lem}
\begin{proof}[Proof]
We note that 
	\begin{eqnarray*}
		W_A(z) &=& {n\choose p}a_pz^p+O(z^p),\\
		W_B(z) &=& {n\choose r}b_rz^r+O(z^r).
	\end{eqnarray*}
We obtain
	\begin{eqnarray*}
		\textrm{for the condition (\ref{war1a}):} \ \ \lim_{z\to 0}U(z)&=&\lim_{z\to 0}\ln\left(\frac{(W_A(z))^s}{(W_B(z))^sz}\right) = -\infty,\\
		\textrm{for the condition (\ref{war1b}):} \ \ \lim_{z\to 0}U(z)&=&\ln\left(\frac{{n \choose p}a_p}{{n \choose r}b_r}\right) < 0.
	\end{eqnarray*}
	
In our one dimensional system the stationary point $\bar{z}>0$ is asymptotically stable iff $\exists_{\varepsilon > 0}\  \dot{z}|_{(\bar{z},\bar{z}+\varepsilon)}<0 \land \dot{z}|_{(\bar{z}-\varepsilon,\bar{z})}>0$. The sufficient condition is $U'(\bar{z})<0$. On the contrary, if $U'(\bar{z})>0$, then the point $\bar{z}$ is unstable. 
Analogously, the sufficient condition for stability of $z=0$ reads: $\lim_{z\to 0}U(z)<0$.
	Since $\sgn(U(z_0))=\sgn(\dot{z}|_{z_0})$, then in the cases (\ref{war1a}) and (\ref{war1b}) we have $\dot{z}<0$ 
for a certain $\varepsilon_z>0$  and $z\in (0,\varepsilon_z)$. From the equality  $\dot{x}=\frac{\dot{z}}{(1+z)^2}$ it follows, that then for a certain $\varepsilon_x>0$ and ~$x\in (0,\varepsilon_x)$ we have $\dot{x}<0$. 

The proof of the second part of Lemma \ref{stab0} is analogous.
\end{proof}

\begin{lem}
	Let $p' = \max{\{ i\colon a_i>0\}},\ r' = \max{\{ i\colon b_i>0\}}$. If 
	\begin{eqnarray*}
		p's>r's+1, &\textrm{ or}&\\
		p's=r's+1 &\textrm{ and }& {n \choose p'}a_{p'}>{n \choose r'}b_{r'},
	\end{eqnarray*}
	then the point $x=1$ is asymptotically stable. If 
	\begin{eqnarray*}
		p's<r's+1, &\textrm{ or}&\\
		p's=r's+1 &\textrm{ and }& {n \choose p'}a_{p'}<{n \choose r'}b_{r'},
	\end{eqnarray*}
	then the point $x=1$ is unstable.
\end{lem}

\begin{proof}[Proof]
It is enough to apply Lemma \ref{stab0} with the strategies interchanged.
\end{proof}
The above theorem guarantees the existence, however not uniqueness of the ME 
(we remind that for the 2-person symmetric games with positive payoffs the uniqueness has been proved for $s \le 1$. Below 
we give an example of a 3-person game with positive payoffs and three ME for the 
sensitivity $s=1$.  

\vskip 0.1cm
{\bf{Example}}

Let $a_0=5,\ a_1=a_2=1,\ b_0=1, \ b_1=\frac{2}{3}, \ b_2=5\frac{1}{3}$. 
The corresponding payoff matrix reads: 
\begin{equation}\label{3personMatrix}
	\begin{array}{r|ccc}
		& 11 & 12 & 22 \\ \hline
		1 & 5 & 1 & 1 \\
		2 & 1 & \frac{2}{3} & 5\frac{1}{3} 		
	\end{array}
\end{equation}
Note that this game can be treated as a coordination game (the players are better off if they play the same strategies).  
The equation for ME reads: 
\begin{equation}
z=\left[\frac{a_0 z^2+a_1 z + a_2}{b_0 z^2+b_1z+b_2}\right]^s, 
\end{equation}
which, for $s=1$ has three positive roots, resulting, after substitution 
$z=\frac{x}{1-x}$ in three ME: 
$x_1=\frac{1}{4}, \ x_2=\frac{1}{2}, \ , x_1=\frac{3}{4}$. 
We checked that $x_1$ and $x_3$ are locally stable, and $x_2$ unstable.

\subsection{Public Goods game}

We apply the general results of the previous subsection to the important 
$(n+1)$--person game--the Public Goods (PG) 
game. In this game each of $(n+1)$ players receives an amount $g$, and chooses one of two actions: C: Contribute with $g$ into the common pool or D: Do not contribute. Let $k$ among $n+1$ players choose C. The amount $kg$ in the common pool is multiplied by $r, \ \ n+1 > r > 1$, and distributed equally among $n+1$ players. 
Using the notation for the payoffs in the $(n+1)$-person game introduced in 
(\ref{wyplatynnn}) and 
the normalization $g=1$ we obtain the following payoffs of the strategies respectively C and D in the PG game: 
\begin{equation} \label{public}
a_k=(k+1)p,\ \ b_k = kp+1, \quad \ \ p:=\frac{r}{n+1}, \quad \frac{1}{n+1}<p<1.
\end{equation}

Below we investigate equilibria for the PG game. We show that their number and stability properties are the same as of a 2-person Prisoner's Dilemma (PD) game, of which it is the multiple, see definition below. Thus, since any 2-person PD with positive payoffs has at least one and at most three internal equilibria, the same is true for the PG game. We begin with the definition of the multiple of any 2-person symmetric game. 
\begin{df}\label{zloz} Multiple of the 2-person symmetric game with the payoff matrix 
	$$
		\left[\begin{array}{cc} a&b\\c&d\end{array}\right]
	$$
	is the $(n\!+\!1)$--person symmetric game with the payoffs 
	$$
		\left.\begin{array}{rcl}a_k&=&ka+(n-k)b,\\
		b_k&=&kc+(n-k)d.\end{array}\right.  
	$$
\end{df}
$k=0,1,...,n.$
Thus, in the multiple of the 2-person symmetric game the payoff of a player is the sum of his payoffs from all the 2-person games with the other $n$ players.  We prove 

\begin{lem}\label{wielo}
	Any 2-person symmetric game and its multiple have the same equilibria in the dynamics 
(\ref{dynnosx}). Moreover, the stability properties of the equilibria of both games are the same. 
\end{lem}
\begin{proof}[Proof]
For the multiple of the 2-person game we have
\begin{eqnarray*}
	W_A(z) &=& \sum_{k=0}^n{n \choose k}(ka+(n-k)b)z^k = \sum_{k=0}^n{n \choose k}k(a-b)z^k + n\sum_{k=0}^n{n \choose k}bz^k \\
    &=& n(a-b)z\sum_{k=1}^n{n-1 \choose k-1}z^{k-1} + nb(1+z)^n\\
	&=& n(a-b)z(1+z)^{n-1} + nb(1+z)^n = n(1+z)^{n-1}(z(a-b) + (1+z)b) \\
    &=& n(1+z)^{n-1}(za+b),\\
\end{eqnarray*}   
and analogously $W_B(z) = n(1+z)^{n-1}(zc+d)$. As previously we define the function $U$, zeros of which are stationary points of the considered dynamics:  
\begin{eqnarray*}
U(z) = s\ln{\frac{W_A(z)}{W_B(z)}} - \ln{z} =  
 s\ln\left(\frac{za+b}{zc+d}\right) - \ln z.
\end{eqnarray*}
Thus, the function $U(z)$ for the $(n+1)$--person game which is the multiple of the 2-person game  is the same as for the 2-person game. In consequence both games have the same equilibria,  and the stability properties of the corresponding equilibria are identical. 
\end{proof}
\noindent Now we come back to the PG $(n\!+\!1)$--person game with the payoffs $a_k=(k+1)p,\ b_k = kp+1$, with $\frac{1}{n+1}<p<1$. We note that they can be rewritten as 
\begin{eqnarray*}
	a_k &=& kp\frac{n+1}{n} + (n-k)p\frac{1}{n}, \\
	b_k &=& k\frac{pn+1}{n} + (n-k)\frac{1}{n}.
\end{eqnarray*}
It means that \emph{Public Goods game} is the multiple of the PD game with the payoff matrix 
\begin{equation}\label{PD}
	\left[\begin{array}{cc} \frac{p(n+1)}{n} & \frac{p}{n}\\
	\frac{pn+1}{n} & \frac{1}{n}\end{array}\right],
\end{equation}
therefore, from Lemma \ref{wielo} it is sufficient to investigate the dynamics for the PD game (\ref{PD}).

\section{Conclusions}

In~this work we discussed the evolutionary dynamics of populations of agents with complex personality profiles, 
which is governed by the attractiveness of strategies rather than by their payoffs.   
The agents can play various types of non-cooperative games, including general two--person asymmetric games and 
multi-person games. The profiles determine the weights the agents associate 
to the payoffs and the popularities of the available strategies. 
It turns out that the polymorphic equilibria and their stability properties are in general determined by a single 
sensitivity to reinforcement parameter which characterizes the population. The parameter has an interpretation in the 
Matching Law of mathematical psychology. 
For general characters there exist stable mixed equilibria of the considered dynamics, not present in the classical 
evolutionary game theory approach based on the replicator dynamics. 

There are various interesting open problems related to the presented research: 
generalization for systems with more behavioral types, introduction of 
other, more general types of attractiveness functions, which for example take into account the speed of the changes 
of the actual composition of population. 
It will also be interesting to allow the actors to change their personalities during the interactions and/or to react with 
delay to received impulses (information lag). In particular, the delay can be 
present only in the social or only in the material part of the attractiveness fuction.  
Models with a finite number of heterogeneous agents with different 
personality characteristics seem to be another interesting area of future research.  

\begin{acknowledgements}
The first author (TP) was supported by~the Polish Government Grant no. N N201 362536.
\end{acknowledgements}

\section{appendix}

\subsection{Ideal personality profiles}

{\bf HE} (Homo Economicus): $\alpha=1$, $ \ \beta=0$. It means that HE assesses attractiveness of the action exclusively through its effectiveness $(u_i=\nu_i)$. HE is interested exclusively in the future prospects,  

{\bf HS} (Homo Sociologicus): $\alpha=0$, $\ \beta=1$. The HS assesses attractiveness of the action only through 
its popularity, the past effect for him/her $(u_i=p_i).$ HS is insensitive to the payoffs of the game.

{\bf HT} (Homo Transcendentalis): $\alpha = \beta =1.$ It is an ideal type, for which every action has the same 
attractiveness $(u_i=1).$ It describes a personality not interested in the effectiveness of a behavior or in its 
propensity, but rather by some other values. Thus, HT is the ideal type insensitive to the payoffs and 
the popularities of the strategies. 

{\bf HA} (Homo Afectualis): $\alpha = \beta =0.$ The attractiveness function takes the form $u_i=\nu_i p_i$ 
The corresponding evolution equations reduce to the standard replicator equations.  
Inserting $\alpha=0$ into (\ref{a2}) we obtain a generalized form of replicator equations, in which the 
material payoffs are included in nonlinear way:
\begin{equation}
\dot p_i = p_i\sum_{j=1,...K} p_j [\nu_i^{1-\beta} - \nu_j^{1-\beta}], \ \ \ i=1,...K.  \label{a11111}
\end{equation}
In particular for $\beta=0$ we obtain the usual replicator equations for the two-person symmetric games with 
$K$ strategies.

\subsection{Symmetric anti-coordination games with 3 strategies}

Here we consider the ~3--strategy game with the payoff matrix 
\begin{equation}\label{3matrix_a}
	\begin{array}{r|ccc}
		& 1 & 2 & 3 \\ \hline
		1 & 0 & a_1 & a_2 \\
		2 & a_1 & 0 & a_3 \\
		3 & a_2 & a_3 & 0
	\end{array}
\end{equation}
$a_i > 0, \ \ i=1,2,3,$ in which the players are better off if they use different 
strategies, then otherwise.  We prove 
\begin{cor}
There exists the unique ME $(\frac{1}{3},\frac{1}{3},\frac{1}{3})$ of the game (\ref{3matrix_a}) with 
$a_i=a>0, \ i=1, 2, 3$ for all $s>0$. 
\end{cor}
\begin{proof}
We check that $x=y=1$, which corresponds to the above ME,  satisfy (\ref{alg1}), (\ref{alg2}). For $x \neq y$, 
subtracting (\ref{alg2}) from (\ref{alg1}) we obtain the equation 
$$x-y=\frac{(1+y)^s-(1+x)^s}{(x+y)^s}, $$
which can not be satisfied due to different signs of both sides. 
\end{proof}
For the general anti--coordination game with arbitrary $a_{ij},  \ s$ such that $a_{ij} > a_{kk} >0$ for all $k$ and all 
$i\neq j, \ \ i, j, k = 1, 2, 3$, 
we checked numerically the existence of an unique ME for all considered numerical values of these parameters.

\subsection{Asymmetric pure coordination games}

\begin{proof}
The dynamics (\ref{dyn2z}) for the pure coordination game 
\begin{equation}
	\left[\begin{array}{cc} (a_1,a_2) & (0,0) \\ (0,0) & (d_1,d_2)\end{array}\right], \label{anticoordmatrix11}
\end{equation} 
with $a_i > 0, \ d_i > 0, \ i=1, 2,$ reads 
\begin{equation}\label{dynkoord}
	\begin{array}{rcl}
	\dot{z_1} &=& z_1^{1-\alpha_1}(1+z_1)^{\alpha_1}\left(\frac{a_1z_2}{1+z_2}\right)^{1-\beta_1} - z_1(1+z_1)^{\alpha_1}\left(\frac{d_1}{1+z_2}\right)^{1-\beta_1}, \\[2mm]
	\dot{z_2} &=& z_2^{1-\alpha_2}(1+z_2)^{\alpha_2}\left(\frac{a_2z_1}{1+z_1}\right)^{1-\beta_2} - z_2(1+z_2)^{\alpha_2}\left(\frac{d_2}{1+z_1}\right)^{1-\beta_2}. 
	\end{array}
\end{equation}
The equations (\ref{zera2z}) for ME of (\ref{dynkoord}) read: 
$\bar{z}_1 = \left(\frac{a_1}{d_1}\bar{z}_2\right)^{s_1}, \ \  
\bar{z}_2 = \left(\frac{a_2}{d_2}\bar{z}_1\right)^{s_2}. $ 
Linearization of (\ref{dynkoord}) around the equilibrium (\ref{stackoord}) leads to the matrix 
$$
	\!\left[\!\begin{array}{l@{\!\!\!\!\!\!\!\!}r}
	-\alpha_1(1+\bar{z}_1)^{\alpha_1}\!\left(\frac{d_1}{1+\bar{z}_2}\right)^{1-\beta_1} &
	(1-\beta_1)\bar{z}_1^{1-\alpha_1}(1+\bar{z}_1)^{\alpha_1}\bar{z}_2^{-\beta_1}\!\left(\frac{a_1}{1+\bar{z}_2}\right)^{1-\beta_1}\\
	(1-\beta_2)\bar{z}_2^{1-\alpha_2}(1+\bar{z}_2)^{\alpha_2}\bar{z}_1^{-\beta_2}\!\left(\frac{a_2}{1+\bar{z}_1}\right)^{1-\beta_2} &
	-\alpha_2(1+\bar{z}_2)^{\alpha_2}\!\left(\frac{d_2}{1+\bar{z}_1}\right)^{1-\beta_2}
	\end{array}\!\right]
$$
with a negative trace.  Using (~\ref{stackoord}) we show that if $s_1s_2<1$ then the determinant of the matrix is positive, 
therefore both eigenvalues have negative real parts, i.e. the equilibrium $(\bar{z}_1,\bar{z}_2)$ is locally 
asymptotically stable. Analogously, if $s_1s_2>1,$ then the determinant is negative, therefore at least one eigenvalue 
has positive real part, i.e.  the equilibrium $(\bar{z}_1,\bar{z}_2)$ is unstable. 
The proof of other statements is omitted. 
\end{proof}

\subsection{Asymmetric anti--coordination games}\label{2chakoord}

We apply the general model for the asymmetric anti--coordination game 
with the payoff matrix 
\begin{equation}
	\left[\begin{array}{cc} (0,0) & (b_1,b_2) \\ (c_1,c_2) & (0,0)\end{array}\right], \label{anticoordmatrix}
\end{equation} 
and $b_i > 0, \ c_i > 0, \ i=1, 2.$ We prove the following theorem
\begin{tw} $ \ $ \label{twistnienie11}

If $s_1s_2\neq 1$ then there exists a unique polymorphic equilibrium of the dynamics (\ref{dyn2z}) for the populations 
with the payoff matrices (\ref{anticoordmatrix}): 
\begin{equation}\label{antistackoord}
	\bar{z}_1 = \left(\frac{b_1c_2^{s_2}}{c_1b_2^{s_2}}\right)^{\frac{s_1}{1-s_1s_2}}, \quad 
    \bar{z}_2 = \left(\frac{c_1^{s_1}b_2}{b_1^{s_1}c_2}\right)^{\frac{s_2}{1-s_1s_2}}.
\end{equation}	
If $s_1s_2 = 1$ then we have the same situation as in Theorem \ref{twistnienie} with obvious changes of the payoff 
parameters. 

The equilibrium (\ref{antistackoord}) is locally asymptotically stable if $s_1s_2<1$, and unstable if $s_1s_2>1$.
\end{tw}
The proof is analogous as in Theorem \ref{twistnienie} and is omitted. 
Numerical simulations indicate that if the equilibrium (\ref{stackoord}) is locally stable then it attracts all investigated 
trajectories from the interior of the simplex of the game.

\subsection{Asymmetric non-coordination games}

We consider populations which play  asymmetric game with the payoff matrices 
\begin{equation}
	\left[\begin{array}{cc} (a,0) & (0,b) \\ (0,c) & (d,0)\end{array}\right]. \label{anticoord}
\end{equation}
This game can be considered as a two-population variant of the Matching Pennies game. The dynamics (\ref{dyn2z}) for the game (\ref{anticoord}), with positive payoffs $a, b, c, d$ reads
\begin{equation}\label{dynakoord}
	\begin{array}{rcl}
	\dot{z_1} &=& z_1^{1-\alpha_1}(1+z_1)^{\alpha_1}\left(\frac{az_2}{1+z_2}\right)^{1-\beta_1} - z_1(1+z_1)^{\alpha_1}\left(\frac{d}{1+z_2}\right)^{1-\beta_1}, \\[2mm]
	\dot{z_2} &=& z_2^{1-\alpha_2}(1+z_2)^{\alpha_2}\left(\frac{c}{1+z_1}\right)^{1-\beta_2} - z_2(1+z_2)^{\alpha_2}\left(\frac{bz_1}{1+z_1}\right)^{1-\beta_2}. 
	\end{array}
\end{equation}
The equations for equilibria have the form 
$\bar{z}_1 = \left(\frac{a\bar{z}_2}{d}\right)^{s_1}, \quad \bar{z}_2 = \left(\frac{b}{c\bar{z}_1}\right)^{s_2}, $
with the solution 
\begin{equation}\label{stacakoord}
	\bar{z}_1 = \left(\frac{ac^{s_2}}{db^{s_2}}\right)^{\frac{s_1}{1+s_1s_2}}, 
\quad \bar{z}_2 = \left(\frac{bd^{s_1}}{ca^{s_1}}\right)^{\frac{s_2}{1+s_1s_2}}.
\end{equation}
We prove that the equilibrium (\ref{stacakoord}) is locally asymptotically stable for all  $s_1,s_2 \in \mathbb{R}_+$.
Linearization of (\ref{dynakoord}) around the stationary point (\ref{stacakoord}) leads to the matrix 
$$
	\!\left[\!\begin{array}{l@{\!\!\!\!}r}
	-\alpha_1(1+\bar{z}_1)^{\alpha_1}\!\left(\frac{d}{1+\bar{z}_2}\right)^{1-\beta_1} &
	(1-\beta_1)\bar{z}_1^{1-\alpha_1}(1+\bar{z}_1)^{\alpha_1}\bar{z}_2^{-\beta_1}\!\left(\frac{a}{1+\bar{z}_2}\right)^{1-\beta_1}\\
	-(1-\beta_2)\bar{z}_2(1+\bar{z}_2)^{\alpha_2}\bar{z}_1^{-\beta_2}\!\left(\frac{b}{1+\bar{z}_1}\right)^{1-\beta_2} &
	-\alpha_2(1+\bar{z}_2)^{\alpha_2}\!\left(\frac{bz_1}{1+\bar{z}_1}\right)^{1-\beta_2}
	\end{array}\!\right]
$$
which has negative trace and positive determinant.


\begin{thebibliography}{99}
\baselineskip=5pt

\bibitem{Veg} Vega-Redondo, F., 2003. Economics and the theory of games. Cambridge, New York: Cambridge University Press
\bibitem{Szabo} G. Szabo, G. Fath \emph {Evolutionary games on graphs}, Physics Reports 446 4-6 (2008) 97-216
\bibitem{McE} McElreath, R., Boyd, R., 2007. Mathematical Models of Social Evolution: 
A Guide for the Perplexed. The University of Chicago Press. Chicago and London. 
\bibitem{Gin1} Gintis, H., 2009. The Bounds of Reason. Game Theory and the Unification of the Behavioral Sciences. 
Princeton University Press. Princeton and Oxford. 
\bibitem{Weibull} Weibull, J. W., 1995. Evolutionary Game Theory. Cambridge: The MIT Press. 
\bibitem{HofSig} J.Hofbauer and K.Sigmund, Evolutionary Games and Population Dynamics,1998, Cambridge University Press
\bibitem{Hen} Henrich, J., Boyd, R., 2001. Why People Punish Defectors. Weak Conformist Transmission can Stabilize Costly Enforcement 
of Norms in Cooperative Dilemmas. J. Theor. Biol. 208, 79-89.
\bibitem{Han} P. J. B. Hancock, L. M. DeBruine (2003) What's a face worth: Noneconomic factors in game playing. In Behavioral and Brain Sciences, 162-163. 
\bibitem{Cob} Cobb, C.W., Douglas, P.H., 1928. A Theory of Production. American Economic Review 18 (Supplement) 139-165.  
\bibitem{Pla1} Platkowski, T., Poleszczuk, J., 2009. Operant Response Theory of Social Interactions. eJourn. Biol. Sci.  1, 1, 1-10. 
\bibitem{Pla2} Platkowski, T., 2010. Cooperation in~Two--Person Evolutionary Games with Complex Personality Profiles, 
J. Theor. Biol. 266 (2010), pp. 522-528
\bibitem{Her} Herrnstein, R.J., 1997. The Matching Law. New York: Harvard University Press. 
 {\bf 8} 4 (2003) pp. 31-38
\bibitem{Pal} Palomino, F.,Vega-Redondo, F., 1999. Convergence of aspirations and (partial) cooperation in the prisoner's dilemma. 
Int. J. Game Theory 28, 465-488 
\bibitem{Pla3} Platkowski, T., Bujnowski, P., 2009. Cooperation in aspiration-based N-person prisoner's dilemmas. Phys. Rev. E 79, 036103   
\bibitem{Gok} C.S.Gokhale, A.Traulsen, \textit{Evolutionary games in~the multiverse}, PNAS 107(2010), 5500--5504. 
\bibitem{Hau3} C. Hauert, F.Michor, M.A. Nowak and M. Doebeli, \textit{Synergy and discounting of cooperation in social dilemmas}, 
J. theor. Biol. 239, 2 (2006) 195-202)  
\bibitem{Pac} Pacheco, J.M., Santos, F., C., Souza, M., O., and Skyrms, B. 2009. Evolutionary dynamics of collective 
action in N-person stag hunt dilemmas. Proc. R. Soc. B 276, 315-321. 
\bibitem{Sou} Souza, M. O., Pacheco, J., M., Santos, F., C., 2009. Evolution of cooperation under N-person snowdrift games, 
J. Theor. Biol. 260(4), 581-588. 
\bibitem{Pla4} Platkowski, T., Zakrzewski, P., 2011. Asymptotically stable equilibrium and limit cycles 
in the Rock--Paper--Scissors game in population of players with complex personalities, Physica A 390 (2011) 4219-4226. 
\bibitem{bro} M. Broom, C. Cannings, G. T. Vickers, \textit{Multi-player matrix games}, 
Bull. Math. Biol., 59 (1997) 931-952
\bibitem{FudNow} L. A. Imhof, D. Fudenberg and M. A. Nowak, Evolutionary cycles of cooperation and defection, 
PNAS, 102, 31 (2005) 10797-10800. 
\end{thebibliography}
\end{document}